\documentclass{amsart}
\usepackage{amsmath,amsthm,amsfonts, amssymb,amscd}
\usepackage[all]{xy}

\newtheorem{theorem}{Theorem}[section]
\newtheorem{lemma}[theorem]{Lemma}
\newtheorem{definition}[theorem]{Definition}
\newtheorem{remark}[theorem]{Remark}
\newtheorem{corollary}[theorem]{Corollary}
\newtheorem{proposition}[theorem]{Proposition}

\newcommand{\minusre}{\hspace{0.3em}\raisebox{0.3ex}{\sl \tiny /}\hspace{0.3em}}
\newcommand{\minusli}{\hspace{0.3em}\raisebox{0.3ex}{\sl \tiny $\setminus $}\hspace{0.3em}}
\newcommand{\lex}{\,\overrightarrow{\times}\,}

\newcommand{\Ker}{\mbox{\rm Ker}}

\newcommand{\Rad}{\mbox{\rm Rad}}

\newcommand{\Infinit}{\mbox{\rm Infinit}}

\newcommand{\RDP}{\mbox{\rm RDP}}

\begin{document}
\title[Lexicographic Product vs Perfect Pseudo Effect Algebras]{Lexicographic Product vs $\mathbb Q$-perfect and $\mathbb H$-perfect Pseudo Effect Algebras}
\author[A. Dvure\v{c}enskij, M. Kola\v{r}\'\i k]{Anatolij Dvure\v censkij$^{1,2}$ and Miroslav Kola\v{r}\'\i k$^2$}
\date{}%
\maketitle
\begin{center}  \footnote{


The paper has been supported by  Slovak Research and Development Agency under the contract APVV-0178-11, the grant VEGA No. 2/0059/12 SAV, and by
CZ.1.07/2.3.00/20.0051.
 }
Mathematical Institute,  Slovak Academy of Sciences,\\
\v Stef\'anikova 49, SK-814 73 Bratislava, Slovakia\\
$^2$ Depart. Algebra  Geom.,  Palack\'{y} University\\
17. listopadu 1192/12, CZ-771 46 Olomouc, Czech Republic\\

E-mail: {\tt
dvurecen@mat.savba.sk} \qquad {\tt miroslav.kolarik@upol.cz}
\end{center}

\begin{abstract}
We study the Riesz Decomposition  Property types of the lexicographic product of two po-groups. Then we apply them to the study of pseudo effect algebras which can be decomposed to a comparable system of non-void slices indexed by some subgroup of real numbers. Finally, we present their representation by the lexicographic product.
\end{abstract}

\section{Introduction}

Quantum structures or quantum logics appeared in the early Thirties when the foundations of a new physics, called now quantum physics, have been founded. It was recognized that the classical rules of statistics of quantum mechanical measurements do not satisfy the axioms of the Kolmogorov probability theory \cite{Kol}, and therefore, it was necessary to find a more general structures than Boolean algebras. The first quantum structures were Boolean algebras, orthomodular lattices and posets, orthoalgebras, D-posets and effect algebras. Effect algebras were introduced  in \cite{FoBe} as a counter part to D-posets introduced in \cite{KoCh}. D-posets have a primary notion difference of two comparable events while effect algebras have a primary notion addition of two mutually excluding events. A prototypical example of effect algebras is the system $\mathcal E(H)$ of all Hermitian operators of a Hilbert space $H$ which are between the zero operator and the identity. In many important examples, an effect algebra is a unit interval in some Abelian partially ordered group. This is true also for $\mathcal E(H)$ because it is an interval in the po-group $\mathcal B(H)$ of all Hermitian operators of $H.$ In \cite{Rav} it was shown that every effect algebra with the Riesz Decomposition Property (RDP in short) is an interval in an Abelian po-group with strong unit and with interpolation. A special subfamily of interval effect algebras are MV-algebras, \cite{CDM}.

In \cite{DvVe1, DvVe2}, the addition $+$ was not more longer assumed to be commutative and, therefore,  pseudo effect algebras were introduced as a noncommutative generalization of effect algebras. Also for these structures it can happen that they are an interval in a unital po-group not necessarily Abelian. This is allowed by a stronger type of RDP, called RDP$_1,$ which in the commutative case is equivalent to RDP.

Perfect effect algebras were introduced in \cite{Dvu2}. For them it is assumed that every element is  either an infinitesimal or it is a co-infinitesimal, i.e. a negation of an infinitesimal.  In \cite{Dv08}, $n$-perfect GMV-algebras, which are pseudo effect algebras with the so-called RDP$_2$ have been introduced. In addition, $n$-perfect pseudo effect algebras were investigated in \cite{DXY} as pseudo effect algebras which can be decomposed into $n+1$ comparable slices. For strong $n$-perfect algebras it was shown that they can be represented by a lexicographic product of a cyclic po-group $\frac{1}{n}\mathbb Z$ with some directed torsion-free po-group with RDP$_1.$

Therefore, the question when the lexicographic product $\mathbb Z \lex G$ has RDP$_1$ was studied in \cite{DvKr} in more detail, and here we present  new results in this direction.

In this paper, we will also study pseudo effect algebras, $E,$ which have a  decomposition $(E_t: t \in [0,1]\cap \mathbb H)$ of $E$ to comparable slices $E_t$'s, where $\mathbb H$ is a subgroup of $\mathbb R$ such that $1 \in \mathbb H.$ We show that in a particular case $\mathbb H= \mathbb Q,$ the group of rational numbers, we are able to present a so-called strong $\mathbb Q$-perfect pseudo effect algebra as an interval in the lexicographic product $\mathbb Q \lex G$ with some directed torsion-free po-group $G.$ The same will be done for a general subgroup $\mathbb H.$ The paper will enlarge the results from \cite{DXY}, where a special case $\mathbb H=\frac{1}{n}\mathbb Z$ was studied as $n$-perfect pseudo effect algebras.

The paper is organized as follows.  Section 2 is gathering the basic notions and results on pseudo effect algebras and po-groups. In Section 3, we extend the Riesz decomposition properties of the lexicographic product of two po-groups studied originally in \cite{DvKr}. Section 4 is concentrated to the description of $\mathbb H$-decompositions of pseudo effect algebras. Section 5 is devoted to the study of $\mathbb H$-perfect pseudo effect algebras. Section 6 describes some representations of strong $\mathbb Q$-perfect pseudo effect algebras by the lexicographic product of type $\mathbb Q\lex G.$ Section 7 is sacrificed  to the most general case~- to strong $\mathbb H$-perfect pseudo effect algebras. In addition, some open problems are formulated.

\section{Basic Facts on Pseudo Effect Algebras and po-groups} 

According to \cite{DvVe1, DvVe2}, we say that a partial algebraic structure
$(E;+,0,1),$  where $+$ is a partial binary operation and 0 and 1 are constants, is said to be a {\it pseudo effect algebra},  if for all $a,b,c \in E,$ the following hold.
\begin{itemize}
\item[{\rm (PE1)}] $ a+ b$ and $(a+ b)+ c $ exist if and only if $b+ c$ and $a+( b+ c) $ exist, and in this case,
$(a+ b)+ c =a +( b+ c).$

\item[{\rm (PE2)}] There are exactly one  $d\in E $ and exactly one $e\in E$ such
that $a+ d=e + a=1.$

\item[{\rm (PE3)}] If $ a+ b$ exists, there are elements $d, e\in E$ such that
$a+ b=d+ a=b+ e.$

\item[{\rm (PE4)}] If $ a+ 1$ or $ 1+ a$ exists,  then $a=0.$ 
\end{itemize}

In every pseudo effect algebra we can derive a partial order, $\le, $ defined as follows:
$a \le b$ if and only if there exists an element $c\in
E$ such that $a+c =b.$ Then  $0$ and $1$ is the least and greatest element of $E.$ In addition,
$a \le b$ if and only if $b = a+c = d+a$ for some $c,d \in E$. We
write $c = a \minusre b$ and $d = b \minusli a.$ Then

$$ (b \minusli a) + a = a + (a \minusre b) = b,
$$
and we write $a^- = 1 \minusli a$ and $a^\sim = a\minusre 1$ for the {\it left} and {\it right negation}, respectively, of any $a \in E.$

We recall that if $+$ is commutative, $E$ is said to
be an {\it effect algebra}; for a comprehensive survey on effect
algebras we recommend  e.g. \cite{DvPu}.

A pseudo effect algebra $E$ is said to be {\it symmetric} if $a^-=a^\sim$ for each $a \in E.$ We note that it can happen that $E$ is symmetric but $+$ is not commutative, see e.g. the proof of Proposition \ref{pr:4.1} below.

A non-empty subset $I$ of a pseudo effect algebra $E$ is said to be an {\it ideal} if (i) $a,b\in I,$ $a+b\in E,$ then $a+b \in I,$ and (ii) if $a\le b \in I,$ then $a \in I.$ An ideal $I\ne E$ is {\it maximal} if it is not a proper subset of any ideal $J \ne E.$

A mapping $h$ from a pseudo effect algebra $E$ into another one $F$ is said to be a {\it homomorphism} if (i) $h(1)=1,$ and (ii) if $a+b$ is defined in $E,$ so is defined $h(a)+h(b)$ and $h(a+b)= h(a)+h(b).$

Let $a$ be an element of a pseudo effect algebra $E$ and $n\ge 0$ be an integer. We define
$$
0a:=0,\quad 1a:=a, \quad na:= (n-1)a +a, \ n \ge 2,
$$
supposing $(n-1)a$ and $(n-1)a+a$ are defined in $E.$ An element $a$ is {\it infinitesimal} if $na$ exists in $E$ for each $n\ge 1.$  We denote by $\Infinit(E)$ the set of all infinitesimal elements of $E.$

A {\it state} on $E$ is any mapping $s: E \to [0,1]$ such that (i) $s(1)=1,$ and (ii) $s(a+b)=s(a)+s(b)$ whenever $a+b$ is defined in $E.$ Let $\mathcal S(E)$ be the set of states on $E.$ It can happen that $\mathcal S(E)$ is empty. A state $s$ is {\it extremal} if from $s=\lambda s_1+(1-\lambda) s_2$ for some $s_1,s_2 \in \mathcal S(E)$ and $0<\lambda <1$ one follows $s_1 = s_2=s.$ 

An ideal $I$ is {\it normal} if $x+I:=\{x+i: i \in I\} = I+x:= \{j+x: j \in I\}$ for any $x \in E.$ For example, if $s$ is a state on $E,$ then
$$
\Ker(s):=\{x \in E: s(x)=0\}
$$
is a normal ideal on $E.$

If $A,B$ are two nonempty subsets of $E,$ $A+B:=\{a+b: a \in A, b \in B, a+b \in E\}.$ We say that $A+B$ is {\it defined} in $E$ if $a+b$ is defined in $E$ for each $a \in A$ and each $b \in B.$  Similarly, we write $A \leqslant B$ if $a\le b$ for each $a \in A$ and each $b \in B.$

We recall that a {\it po-group} (= partially ordered group) is a
group $(G;+,0)$ endowed with a partial order $\le$ such that if $a\le b,$ $a,b
\in G,$ then $x+a+y \le x+b+y$ for all $x,y \in G.$  We denote by
$G^+=\{g \in G: g \ge 0\}$ the {\it positive cone} of $G.$ If, in addition, $G$
is a lattice under $\le$, we call it an $\ell$-group (= lattice
ordered group).  An element $u\in G^+$ is said to be a {\it strong unit}
(= order unit) if given $g \in G,$ there is an integer $n\ge 1$ such
that $g \le nu,$ and the couple $(G,u)$ with a fixed strong unit $u$ is
said to be a {\it unital po-group.}

The set $C(G):=\{c\in G: c+g=g+c, \forall g \in G\}$ is said to be the {\it commutative center} of $G.$

We say   a po-group $G$ is {\it directed} if, given $g_1,g_2 \in G,$ there is $g \in G$ such that $g_1,g_2 \le g.$ This is equivalent to the statement: given $h_1,h_2 \in G,$ there is $h \in G$ such that $h\le h_1,h_2.$ A subgroup $H$ of a po-group $G$ is {\it convex} if $a,b\in H$ and $c \in G$ with $a\le c \le b$ imply $c \in H.$

An important class of pseudo effect algebras can be obtained as follows. If $(G,u)$ is a unital (not necessary Abelian) po-group
with strong unit $u$, and
$$
\Gamma(G,u) := \{g \in G: \ 0 \le g \le u\},\eqno(2.1)
$$
then $(\Gamma(G,u); +,0,u)$ is a pseudo effect algebra if we
restrict the group addition $+$ to $\Gamma(G,u).$  Every pseudo
effect algebra $E$ that is isomorphic to some $\Gamma(G,u)$ is said
to be an {\it interval pseudo effect algebra}.

For basic properties of pseudo effect algebras see \cite{DvVe1} and
\cite{DvVe2}.

The following kinds of the Riesz Decomposition properties were introduced in \cite{DvVe1,DvVe2} and are crucial for the study of pseudo effect algebras.

We say that a  po-group $G$ satisfies

\begin{enumerate}
\item[(i)]
the {\it Riesz Interpolation Property} (RIP for short) if, for $a_1,a_2, b_1,b_2\in G,$  $a_1,a_2 \le b_1,b_2$  implies there exists an element $c\in G$ such that $a_1,a_2 \le c \le b_1,b_2;$

\item[(ii)]
\RDP$_0$  if, for $a,b,c \in G^+,$ $a \le b+c$, there exist $b_1,c_1 \in G^+,$ such that $b_1\le b,$ $c_1 \le c$ and $a = b_1 +c_1;$

\item[(iii)]
\RDP\  if, for all $a_1,a_2,b_1,b_2 \in G^+$ such that $a_1 + a_2 = b_1+b_2,$ there are four elements $c_{11},c_{12},c_{21},c_{22}\in G^+$ such that $a_1 = c_{11}+c_{12},$ $a_2= c_{21}+c_{22},$ $b_1= c_{11} + c_{21}$ and $b_2= c_{12}+c_{22};$

\item[(iv)]
\RDP$_1$  if, for all $a_1,a_2,b_1,b_2 \in G^+$ such that $a_1 + a_2 = b_1+b_2,$ there are four elements $c_{11},c_{12},c_{21},c_{22}\in G^+$ such that $a_1 = c_{11}+c_{12},$ $a_2= c_{21}+c_{22},$ $b_1= c_{11} + c_{21}$ and $b_2= c_{12}+c_{22}$, and $0\le x\le c_{12}$ and $0\le y \le c_{21}$ imply  $x+y=y+x;$

\item[(v)]
\RDP$_2$  if, for all $a_1,a_2,b_1,b_2 \in G^+$ such that $a_1 + a_2 = b_1+b_2,$ there are four elements $c_{11},c_{12},c_{21},c_{22}\in G^+$ such that $a_1 = c_{11}+c_{12},$ $a_2= c_{21}+c_{22},$ $b_1= c_{11} + c_{21}$ and $b_2= c_{12}+c_{22}$, and $c_{12}\wedge c_{21}=0;$

\item[(vi)]
the {\it Strong Riesz Interpolation Property} ({\rm SRIP} for short)  if, for $a_1,a_2, b_1, b_2 \in G^+$ with $a_1,a_2 < b_1, b_2,$ there is $c \in G$ such that $a_1,a_2 < c < b_1,b_2.$
\end{enumerate}

If, for $a,b \in G^+,$ we have for all $0\le x \le a$ and $0\le y\le b,$ $x+y=y+x,$ we denote this property by $a\, \mbox{\rm \bf com}\, b.$

For Abelian po-groups, RDP, RDP$_1,$ RDP$_0$ and RIP are equivalent.

We recall that the po-groups of reals numbers $\mathbb R$ and rational numbers $\mathbb Q$ satisfy SRIP, but the group of integers $\mathbb Z$ does not. In addition, let $\mathbb H$ be any subgroup of $\mathbb R.$  Due to \cite[Lem 4.21]{Goo}, $\mathbb H$ is either cyclic, i.e. it is isomorphic to some $\frac{1}{n}\mathbb Z$ or is dense in $\mathbb R.$ Also in the latter case, SRIP holds for $\mathbb H.$

A po-group $G$ satisfies \RDP$_2$ iff $G$ is an $\ell$-group, \cite[Prop 4.2(ii)]{DvVe1}.

The RDP will be denoted by the following table:

$$
\begin{matrix}
a_1  &\vline & c_{11} & c_{12}\\
a_{2} &\vline & c_{21} & c_{22}\\
  \hline     &\vline      &b_{1} & b_{2}
\end{matrix}\ \ .
$$

We say that a pseudo effect algebra $E$ satisfies the above types of the Riesz decomposition properties, if in the definition of RDP's, we change $G^+$ to $E.$

The basic result of pseudo effect algebras is the following representation theorem \cite[Thm 7.2]{DvVe2}:

\begin{theorem}\label{th:2.1'}
For every pseudo effect algebra with \RDP$_1,$ there is a unique (up to isomorphism of unital po-groups) unital po-group $(G,u)$ with \RDP$_1$\ such that $E \cong \Gamma(G,u).$
\end{theorem}

We note that a  pseudo effect algebra $E$ satisfies RDP$_2$ iff $E$ is a lattice and if it satisfies RDP$_1$ iff $E$ is in fact a pseudo MV-algebra, \cite[Thm 8.8]{DvVe2}. For them we have that the variety of pseudo MV-algebras is categorically equivalent to the category of unital $\ell$-groups, see \cite{151}.

A {\it state} on a unital po-group $(G,u)$ is any mapping $s:G \to \mathbb R$ such that (i) $s(g)\ge 0$ for $g\in G^+,$ (ii) $s(g_1+g_2)=s(g_1)+s(g_2),$ $g_1,g_2 \in G,$ and (iii) $s(1)=1.$ If a pseudo effect algebra $E=\Gamma(G,u)$ satisfies RDP$_1,$ then every state on $E$ can be extended to a unique state on $(G,u),$ and conversely, the restriction of every state on $(G,u)$ to $[0,u]$ gives a state on $E,$ see \cite{DvVe2}. In addition, an extremal state is extendible only to an extremal state on $(G,u)$ and vice-versa.

\section{Lexicographic product and RDP's}

In this section, we continue with the study of RDP's properties of the lexicographic product which started in \cite{DvKr} and we concentrate to the lexicographic product of two po-groups.

Let $G_1$ and $G_2$ be two po-groups and we define the  direct product group $G_1 \times G_2.$ We define the {\it lexicographic order} $\le$ on  $G_1 \times G_2$ by $(g_1,g_2) \le (h_1,h_2)$ iff either $g_1 < g_2$ or $g_1=g_2$ and $h_1 \le h_2,$ for $(g_1,g_2), (h_1,h_2) \in G_1 \times G_2,$ and $G_1 \lex G_2$ will denote the {\it lexicographic product} of $G_1$ and $G_2$ endowed with the above defined lexicographic order.

\begin{theorem}\label{th:2.1} Let $G_1$ be a linearly ordered po-group and let $G_2$ be a directed po-group with $\RDP.$   Then the po-group $G=G_1 \lex G_2$ satisfies \RDP.
\end{theorem}

\begin{proof}
Let $G_2$ be a directed po-group with RDP.
The positive cone  $(G_1 \lex G_2)^+$ is of the form
$\{(0,a): a \in G_2^+\} \cup \{(n,a): n \in G_1^+\setminus \{0\},\ a \in G_2\}.$
Assume that
$$
(m_1,a_1)+(m_2,a_2) = (n_1,b_1) + (n_2,b_2)$$
holds in $(G_1 \lex G_2)^+.$

(i) Let $(0,a_1)+(0,a_2) = (0,b_1)+ (0,b_2).$  Then $a_1,a_2,b_1,b_2 \in G_2^+$ and RDP for this case  follows from RDP for $G_2.$

(ii) Let  $m_1+m_2 >0$ and $m_2=0.$ For the subcase $n_2 >0,$  we use the same ideas as  case (i) of \cite[Prop 3.3, Thm 3.10]{DvKr}, and we have the following decomposition table

$$
\begin{matrix}
(m_1,a_{1})&\vline  &  (n_1,b_1) & (n_2,-b_1+a_1)\\
(0,a_{2}) &\vline  & (0,0) & (0,a_2) \\
\hline &\vline &(n_1,b_{1}) & (n_2,b_{2})
\end{matrix} \ \ \ \mbox{for}\ n_2 >0.
$$

For the subcase $n_2 = 0,$ we have $a_2,b_2 \ge 0.$ 
As $G_2$ is directed, there exists $d \in G_2$ with $d \le a_1,b_1.$  Then there exist $e_{11},e_{12},e_{21},e_{22} \in G_2^+$ such that

$$
\begin{matrix}
-d+ a_{1} &\vline &  e_{11} & e_{12}\\
a_{2} &\vline  & e_{21} & e_{22} \\
       \hline     &\vline    &-d+ b_{1} & b_{2}
\end{matrix}\ \ ,
$$
which yields

$$
\begin{matrix}
a_{1} &\vline & d+ e_{11} & e_{12}\\
a_{2} &\vline  & e_{21} & e_{22} \\
\hline &\vline  &b_{1} & b_{2}
\end{matrix}\ \ ,
$$
and finally

$$
\begin{matrix}
(m_1,a_{1}) &\vline  &  (n_1,d+e_{11}) & (0,e_{12})\\
(0,a_{2}) &\vline   & (0,e_{21}) & (0,e_{22}) \\
\hline &\vline &(n_1,b_{1}) & (0,b_{2})
\end{matrix} \ \ \ \mbox{for}\ n_2 =0.
$$

The case $(m_1,a_1)+(0,a_2)=(0,b_1)+(n_2,b_2)$ follows from the first case of (ii).

In a similar way we deal with the case $(0,a_1)+(m_2,a_2) = (n_1,b_1)+ (n_2,b_2),$ where  $n_1 \ge 1.$ Then $m_2 \ge 1,$  $a_1 \ge 0,$ and we use the decomposition

$$
\begin{matrix}
(0,a_{1}) &\vline  &  (0,a_1) & (0,0)\\
(m_2,a_{2}) &\vline  & (n_1,-a_1+b_1) & (n_2,b_2) \\
\hline &\vline &(n_1,b_{1}) & (n_2,b_{2})
\end{matrix}\ \ .
$$

(iii) Let $m_1,m_2,n_1,n_2> 0.$ We use the tables

$$
\begin{matrix}
-d+a_{1} &\vline &  e_{11} & e_{12}\\
a_{2}-d &\vline  & e_{21} & e_{22} \\
\hline &\vline &-d+ b_{1} & b_{2}-d
\end{matrix}\ \ ,
$$
where $d \le a_1,a_2,b_1,b_2,$ and then
$$
\begin{matrix}
(m_1,a_{1}) &\vline &  (n_1,d+e_{11}) & (-n_1+m_1,e_{12})\\
(m_2,a_{2}) &\vline  & (0,e_{21}) & (m_2,e_{22}+d) \\
\hline &\vline &(n_1,b_{1}) & (n_2,b_{2})
\end{matrix} \ \ \ \mbox{for}\  m_1 \ge n_1,
$$
where $d \le a_1,a_2,b_1,b_2,$ and finally
$$
\begin{matrix}
(m_1,a_{1}) &\vline &  (m_1,d+e_{11}) & (0,e_{12})\\
(m_2,a_{2}) &\vline  & (-m_1+n_1,e_{21}) & (n_2,e_{22}+d) \\
\hline &\vline &(n_1,b_{1}) & (n_2,b_{2})
\end{matrix}\ \ \ \mbox{for}\  n_1 \ge m_1,
$$
where $d \le a_1,a_2,b_1, b_2.$ The second case follows also from the first one, if we take into account that $(n_1,b_{1}) + (n_2,b_{2})= (m_1,a_1)+(m_2,a_2).$
\end{proof}

\begin{theorem}\label{th:2.2} Let $G_1$ be a linearly ordered po-group and let $G_2$ be an Abelian directed po-group with $\RDP.$   Then the po-group $G_1 \lex G_2$ satisfies \RDP$_1.$
\end{theorem}

\begin{proof}
We use the same decompositions as in the proof of Theorem \ref{th:2.1}. It is enough to verify RDP$_1$ for
(iii) of Theorem \ref{th:2.1}. Assume that $(0,0)\le (x_1,x_2)\le (0,e_{21})$ and $(0,0)\le (y_1,y_2) \le (-n_1+m_1,e_{12}),$ then $x_1=0$ and evidently, $(x_1,x_2)+(y_1,y_2)=(y_1,y_2)+(x_1,x_2).$
\end{proof}


The following result has been proved in \cite[Thm 3.7]{DvKr}.

\begin{theorem}\label{th:2.3} A po-group   $\mathbb{Z} \lex G$  satisfies $\RDP_1$ if and only if $G$ is a directed po-group with  $\RDP_1.$

In addition, there is a categorical equivalence of the category of pseudo effect algebras with \RDP$_1$ with the category of unital po-groups $(G,u)$ with \RDP$_1.$
\end{theorem}

Let $n \ge 1$ be an integer. We define $\frac{1}{n} \mathbb Z$ as the $\ell$-group of rational numbers of the form $\{\frac{i}{n}: i \in \mathbb Z\}$ which is an $\ell$-group isomorphic to the group of integers $\mathbb Z.$ The number $1$ is a strong unit for $\frac{1}{n} \mathbb Z.$ 

\begin{proposition}\label{pr:2.4}
Let $n\ge 1$ be an integer and $G$ a po-group. The following statements are equivalent:

\begin{itemize}

\item[{\rm (i)}] The po-group $\frac{1}{n} \mathbb Z \lex G$ satisfies \RDP$_1.$

\item[{\rm (ii)}]  The po-group $G$ is directed and  with \RDP$_1.$

\item[{\rm (iii)}] The po-group $G$ is directed and the pseudo effect algebra $\Gamma(\frac{1}{n} \mathbb Z \lex G,(1,0))$ satisfies \RDP$_1.$
\end{itemize}

\end{proposition}

\begin{proof} (i) $\Leftrightarrow$ (ii) The po-group $E_n:=\frac{1}{n} \mathbb Z \lex G$ is isomorphic to the po-group  $\mathbb Z \lex G.$ By \cite[Thm 3.7]{DvKr}, $ \mathbb Z \lex G$ satisfies RDP$_1$ iff $G$ is a directed po-group with RDP$_1.$ Consequently, the po-group $\frac{1}{n} \mathbb Z \lex G$ satisfies RDP$_1$ iff $G$ is a directed po-group with RDP$_1.$

(i) $\Rightarrow$ (iii) is trivial.

Finally,  let (iii) hold. Assume $ g_1,g_2 \le h_1,h_2$ for some $g_1,g_2,h_1,h_2 \in G.$ Directness of $G$ entails the existence of $g_0\in G$ such that $g_0 \le g_1,g_2.$  Then $0\le g_1-g_0,g_2-g_0 \le h_1-g_0,h_2-g_0$ which yields
$(0,g_1-g_0),(0,g_2-g_0) \le (0,h_1-g_0),(0,h_2-g_0) \in E_n$ and RDP$_1$ of $E_n$ implies the existence of $g \in G^+$ such that $(0,g_1-g_0),(0,g_2-g_0) \le (0,g)\le (0,h_1-g_0),(0,h_2-g_0).$  Hence, $g_1-g_0, g_2-g_0 \le g \le h_1-g_0, h_2 - g_0$ and finally, $g_1,g_2 \le  g+g_0 \le h_1, h_2$ which implies $G$ satisfies RIP. By
\cite[Prop 3.1(i)]{DvKr}, $\frac{1}{n}\mathbb Z\lex G$ satisfies RIP,  and by \cite[Thm 3.6]{DvKr}, $\frac{1}{n}\mathbb Z\lex G$ satisfies RDP$_1.$
\end{proof}

\begin{theorem}\label{th:2.5}
Let $G$ be a po-group. Define the following statements:

\begin{itemize}

\item[{\rm (i)}] The po-group $ \mathbb Q \lex G$ satisfies \RDP$_1.$

\item[{\rm (ii)}]  The po-group $G$ is    with \RDP$_1.$

\item[{\rm (ii')}]  The po-group $G$ is directed and  with \RDP$_1.$

\item[{\rm (iii)}] The pseudo effect algebra $E:=\Gamma( \mathbb Q \lex G,(1,0))$ satisfies \RDP$_1.$

\end{itemize}
Then {\rm (i)} implies {\rm (ii)}, and {\rm (i)} and {\rm (iii)} are equivalent. If $G$ is directed, then {\rm (i)} and {\rm (ii')} are equivalent, too.
\end{theorem}

\begin{proof}
(i) $\Rightarrow$ (ii). By \cite[Prop 3.2 (2)]{DvKr}, $G$ is a po-group with RDP$_1.$

(i) $\Rightarrow$ (iii). It is trivial.

(iii) $\Rightarrow$ (i). Since $G$ satisfies also RIP, by \cite[Prop 3.2 (1)]{DvKr}, $\mathbb Q \lex G$ satisfies RIP. Applying \cite[Thm 3.6]{DvKr}, $\mathbb Q\lex G$ satisfies RDP$_1.$

Now let $G$ be directed. Due to the implication (i) $\Rightarrow$ (ii), (i) implies (ii'). Conversely, let $G$ satisfy RDP$_1,$ and let for $(r_1,g_1),(r_2,g_2), (q_1,h_1), (q_2,h_2) \in (\mathbb Q \lex G)^+,$ we have $(r_1,g_1)+(r_2,g_2)= (q_1,h_1)+(q_2,h_2).$ There is an integer $n \ge 2$ such that $r_1 =\frac{i_1}{n},$ $r_2 = \frac{i_2}{n},$ $q_1 =\frac{j_1}{n}$ and $q_2 = \frac{j_2}{n}.$ Applying Proposition \ref{pr:2.4}(i), $\frac{1}{n}\mathbb Z \lex G$ satisfies RDP$_1.$ Hence, we can found the necessary RDP$_1$ decomposition for $(\frac{i_1}{n},g_1)+(\frac{i_2}{n},g_2)= (\frac{j_1}{n},h_1)+(\frac{j_2}{n},h_2),$ which implies that $\mathbb Q \lex G$ satisfies RDP$_1$
\end{proof}

We note that in contrast to (ii) of Proposition \ref{pr:2.4}, $G$ is not necessarily directed because $\mathbb Q$ satisfies SRIP.


\vspace{3mm}

The following result can be proved in the same way as Theorem \ref{th:2.5}.

\begin{theorem}\label{th:2.6}
Let $G$ be a po-group. Define the following statements:

\begin{itemize}

\item[{\rm (i)}] The po-group $ \mathbb R\lex G$ satisfies \RDP$_1.$

\item[{\rm (ii)}]  The po-group $G$ is    with \RDP$_1.$

\item[{\rm (iii)}] The pseudo effect algebra $E:=\Gamma( \mathbb R\lex G,(1,0))$ satisfies \RDP$_1.$

\end{itemize}
Then {\rm (i)} and {\rm (iii)} are equivalent, and {\rm (i)} implies {\rm (ii)}.
\end{theorem}

Now we show when (ii) in Theorem \ref{th:2.6} implies (i).


\begin{theorem}\label{th:2.7}  Let $G$ be a directed po-group. The following statements are equivalent:
\begin{itemize}

\item[{\rm (i)}] The po-group $ \mathbb R\lex G$ satisfies \RDP$_1.$

\item[{\rm (ii)}]  The po-group $G$ is    with \RDP$_1.$

\item[{\rm (iii)}] The pseudo effect algebra $E:=\Gamma( \mathbb R\lex G,(1,0))$ satisfies \RDP$_1.$

\end{itemize}
\end{theorem}

\begin{proof}
According to Theorem \ref{th:2.6}, it is necessary to show that (ii) implies (i).  Assume that $(s_1,a_1)+(s_2,a_2)= (t_1,b_1) + (t_2,b_2)$ holds in $(\mathbb R \lex G)^+.$

If we use the RDP decomposition tables from the proof of Theorem \ref{th:2.1} which were used there for two first cases, we see that for theses particular cases we have RDP$_1$ decompositions. According to that proof, we can assume that all $s_1, s_2, t_1, t_2$ are strictly positive. Since $\mathbb Q$ is dense in $\mathbb R$, we can found two rational numbers $r_1$ and $r_2$ in $\mathbb Q^+$ such that $0<r_i <s_i$ if $s_i$ is irrational, otherwise, $r_i = s_i.$ We can choose also two additional rational numbers $q_1,q_2 \in \mathbb Q$ such that $q_1+q_2 =r_1+r_2$ and $0<q_i< t_i$ if $t_i>0$ otherwise $q_i = t_i.$ 
By Theorem \ref{th:2.5}, $ \mathbb Q  \lex G$ satisfies RDP$_1,$ therefore, we can found the following RDP$_1$ decomposition table

$$
\begin{matrix}
(r_1,a_{1}) &\vline &  (r_{11},e_{11}) & (r_{12},e_{12})\\
(r_2,a_{2}) &\vline  & (r_{21},e_{21}) & (r_{22},e_{22}) \\
\hline &\vline &(q_1,b_{1}) & (q_2,b_{2})
\end{matrix}\ \ ,
$$
where $(r_{ij},e_{ij})  \in (\mathbb Q \lex G)^+$ for $i,j =1,2.$  Without loss of generality, we can assume that all $r_{ij}$'s are strictly positive.

Since $(s_1-r_1)+(s_2-r_2)=(t_1-q_1)+(t_2-q_2)$ and all summands are nonnegative,  using RDP holding in $\mathbb R$, we have the RDP decomposition table

$$
\begin{matrix}
(s_1 -r_1) &\vline &  s_{11} & s_{12}\\
(s_2 - r_2) &\vline  & s_{21} & s_{22} \\
\hline &\vline &(t_1- q_{1}) & (t_2- q_{2})
\end{matrix}\ \ ,
$$
where $s_{11},s_{12},s_{21},s_{22} \in [0,\infty).$ Then we have the following RDP decomposition

$$
\begin{matrix}
(s_1,a_{1}) &\vline &  (r_{11} +s_{11},e_{11}) & (r_{12}+s_{12},e_{12})\\
(s_2,a_{2}) &\vline  & (r_{21}+s_{21},e_{21}) & (r_{22}+s_{22},e_{22}) \\
\hline &\vline &(t_1,b_{1}) & (t_2,b_{2})
\end{matrix}\ \ .
$$
Finally, we have to show that   $(r_{12}+s_{12},e_{12})    \,\mbox{\rm \bf com}\, (r_{21}+s_{21},e_{21}).$ Choose $(0,0)\le (x_1,x_2) \le (r_{12}+s_{12},e_{12})$ and $(0,0)\le (y_1,y_2) \le (r_{21}+s_{21},e_{21}).$

We have the following four cases.

(i) $x_1 >0,$ $y_1>0.$ Choose two rational numbers $r'$ and $r''$ such that $0<r'<x_1$ and $0<r''<y_1.$ Then $(0,0)< (r',x_2)< (x_1,x_2),$ $(0,0) < (r'',y_2) < (y_1,y_2),$ and $(0,0)\le (r',x_2) < (r_{12},e_{12}),$ $(0,0)\le (r'',y_2)< (r_{21}, e_{21}).$ By the RDP$_1$ table used above and holding in $\mathbb Q \lex G,$   $(r',x_2)$ and $(r'',y_2)$ commute, so that $x_2$ and $y_2$ commute, which yields $ (x_1,x_2) + (y_1,y_2) = (y_1,y_2)+ (x_1,x_2).$

(ii) $x_1 =0,$ $y_1=0.$ Then $(0,0) \le (0,x_2) < (r_{12}+s_{12},e_{12})$ which entails $(0,x_2) < (r_{12},e_{12})$ and similarly, $(0,0)\le (0,y_2) < (r_{21},e_{21}).$  Therefore, we have $(0,x_2)$ and $(0,y_2)$ commute, so that $x_2$ and $y_2$ commute, too, and finally $(x_1,x_2)$ and $(y_1,y_2)$ commute.

(iii) $x_1 >0,$ $y_1=0.$  Choose a rational number $r >0$ such $0<r<x_1.$ Then $(0,0)\le (r,x_2) \le (x_1,x_2) \le (r_{12}+s_{12},e_{12}).$ Then $(0,0) \le (r,x_2)< (r_{12},e_{12}).$ Similarly as in the case (ii), we have $(0,0)\le (0,y_2) < (r_{21},e_{21}).$ Therefore, $(x,x_2)$ and $(0,y_2)$ commute, which entails $x_2$ and $y_2$ also commute, and finally $(x_1,x_2)$ and $(y_1,y_2)$ commute.

(iv) $x_1 =0,$ $y_1>0.$ It follows the same steps as the proof of the case (iii).

Combining all cases (i)--(iv), we see that RDP$_1$ holds as was claimed.
\end{proof}

The last two theorems can be generalized as follows.

\begin{theorem}\label{th:2.8}  Let $G$ be a directed po-group and $\mathbb H$ be a subgroup of the group $\mathbb R$ such that $1\in \mathbb H.$ The following statements are equivalent:
\begin{itemize}

\item[{\rm (i)}] The po-group $ \mathbb H\lex G$ satisfies \RDP$_1.$

\item[{\rm (ii)}]  The po-group $G$ is    with \RDP$_1.$

\item[{\rm (iii)}] The pseudo effect algebra $E:=\Gamma( \mathbb H\lex G,(1,0))$ satisfies \RDP$_1.$

\end{itemize}
\end{theorem}

\begin{proof}
(i) $\Rightarrow$ (ii). It follows from \cite[Prop 3.2(2)]{DvKr}.

(i) $\Rightarrow$ (iii). It is trivial.

(iii) $\Rightarrow$ (i). Since $G$ satisfies also RIP, by \cite[Prop 3.2 (1)]{DvKr}, $\mathbb H \lex G$ satisfies RIP. Applying \cite[Thm 3.6]{DvKr}, $\mathbb H\lex G$ satisfies RDP$_1.$

(ii) $\Rightarrow$ (i). By the assumption, $G$ is directed. Due to \cite[Lem 4.21]{Goo}, $\mathbb H$ is either a cyclic subgroup of $\mathbb R$ or a dense subset of $\mathbb R.$

In the first case, $\mathbb H=\frac{1}{n}\mathbb Z$ for some integer $n\ge 1.$ Therefore, we can apply Proposition \ref{pr:2.4}, and whence, $\mathbb H \lex G$ has RDP$_1.$

In the second case, we can literally  apply the proof of the implication (ii) $\Rightarrow$ (i) of Theorem \ref{th:2.7} changing $\mathbb R$ to $\mathbb H.$
\end{proof}

\vspace{2mm}
{\bf Problem 1.}
Since both $\mathbb Q,$  $\mathbb R$ as well as any non-discrete $\mathbb H$  satisfy SRIP, according to \cite[Cor 2.12]{Goo} or \cite[Prop 3.2(1)]{DvKr},
$\mathbb Q \lex G,$  $\mathbb R \lex G$ and $\mathbb H \lex G$ satisfy RDP$_1$ whenever $G$ is an Abelian po-group, not necessarily directed. It is unknown whether is this true also for any po-group $G$ not necessarily Abelian.


\section{$\mathbb H$-decompositions of Pseudo Effect Algebras}

In the rest of the paper we will assume that $\mathbb H$ is a  subgroup of the group of real numbers $\mathbb R$ such that $1 \in \mathbb H$ and we denote $[0,1]_\mathbb H:= [0,1]\cap \mathbb H.$ We define two types of decompositions of a pseudo effect algebra indexed by elements of $[0,1]_\mathbb H.$

We say that a decomposition $(E_t: t \in [0,1]_\mathbb H)$ of a pseudo effect algebras $E,$ i.e. a system $(E_t: t \in [0,1]_\mathbb H)$ of nonempty subsets of $E$ such that  $E_s \cap E_t = \emptyset$ for $s< t,$ $s,t \in [0,1]_\mathbb H$ and $\bigcup_{t \in [0,1]_\mathbb H} E_t = E$ is

\begin{enumerate}

\item[(I)] an $\mathbb H$-{\it decomposition of type I} if
\begin{enumerate}

\item[(Ii)] $E_{s}+ E_{t}$ exists   if $s+t<1,$

\item[(Iii)] $E_{0}$ is a unique maximal ideal of $E;$
\end{enumerate}

\item[(II)] an  $\mathbb H$-{\it decomposition} if

\begin{enumerate}


\item[{\rm (a)}] $E_{t}^{-}=E_{t}^{\sim}=E_{1-t}$ for any $t\in
[0,1]_\mathbb H,$

\item[{\rm (b)}] if $x\in E_{s},$ $y\in E_{t}$ and $x+ y$ exists in
$E$, then $s+t\le 1$ and $x+ y\in E_{s+t}$ for
$s,t \in [0,1]_\mathbb H.$
\end{enumerate}

\end{enumerate}

For example, if $E=\Gamma(\mathbb R \lex G, (1,0)),$ then  $(E_t: t \in [0,1]),$ where $E_0:=\{(0,g): g \in G^+\},$ $E_t =\{(t,g): g \in G\}$ for $0<t<1$ and $E_1 =\{(1,-g): g \in G^+\},$ and $G$ is a po-group, is an $\mathbb R$-decomposition of the pseudo effect algebra $E$ of both types. In fact, we have $E_s+E_t = E_{s+t},$ whenever $s+t <1.$ To show that $E_0$ is a maximal ideal, take $(t,g) \in E\setminus E_0,$ and let $I$ be the ideal of $E$ generated by $E_0$ and $(t,g).$ Then, for any $s<t,$ $E_s \subseteq I.$ In particular, there is an integer $n$ such that $1/n <t.$  Then $(1/n,0) \in I$ but $(1,0)=n(1/n,0) \in I$ proving $I=E.$ In addition, it is easy to verify that $E_0$ is also a normal ideal,  $E_0 = \Infinit(E)$ and $E_t^-=E_{1-t}=E_t^\sim$ for any $t \in [0,1],$ and $E_0$ is a unique maximal ideal of $E.$

We say that a state $s$ on a pseudo effect algebra $E$ is an $\mathbb H$-{\it valued state} if $s(E)=[0,1]_\mathbb H,$ where $[0,1]_\mathbb H :=\Gamma(\mathbb H,1)= [0,1]\cap \mathbb H.$ If $s(E) \subseteq [0,1]_\mathbb H,$ we say that $s$ is an $\mathbb H$-{\it state}.  In particular, if $\mathbb H = \frac{1}{n}\mathbb Z,$  a $\frac{1}{n}\mathbb Z$-valued state is also said to be an $(n+1)$-{\it valued discrete state}, \cite{DXY}.

\begin{theorem}\label{th:3.1}
Let  $E$ be a pseudo effect algebra and $\mathbb H$ be a subgroup of $\mathbb R$ containing $1$. The following two statements
are equivalent:
\begin{enumerate}
\item[{\rm (i)}] There exists an $\mathbb H$-valued state on $E$.

\item[{\rm (ii)}] There exists an $\mathbb H$-decomposition $(E_t: t \in [0,1]_\mathbb H)$ of nonempty  subsets of $E.$
\end{enumerate}
\end{theorem}

\begin{proof}
Let $s$ be an $\mathbb H$-valued state. Given $t \in [0,1]_\mathbb H,$ we define $E_t:= s^{-1}(\{t\}).$ It is evident that,  the system $(E_t: t \in [0,1]_\mathbb H),$ is a decomposition of $E.$ For (a), let $x \in E_t.$ Then $s(x)=t$ and $s(x^-)=1-s(x)=s(x^\sim),$ which proves $x^-, x^\sim \in E_{1-t}.$ Conversely, if $y \in E_{1-t},$ then $y^-,y^\sim \in E_t.$ For (b), assume $x \in E_s$ and $y \in E_t$ and let $x+y$ be defined in $E.$ Then $s(x+y) = s(x)+s(y)=s+t\le 1,$ which implies $x+y \in E_{s+t}.$ Then $(E_t: t \in [0,1]_\mathbb H)$ is an $\mathbb H$-decomposition of $E.$

Conversely, let (ii) hold. Define a mapping $s: E \to [0,1]$ by $s(x)=t$ iff $x \in E_t.$ Take $x,y \in E$ such that $s(x)= t_1$ and $s(y)=t_2.$ Since $x \in E_{t_1}$ and $y \in E_{t_2},$ due to (b), we have $t_1+t_2 \le 1.$ Hence, $s(x+y)=s(x)+s(y).$ There exists a unique $t \in [0,1]_\mathbb H$ such that $0\in E_t.$ For every $x \in E_1,$ $x+0=x=0+x,$ thus by (b), $t+1 \le 1$ which yields $t = 0$ and therefore, $s(0)=0$ and $s(1)=1.$ In other words, $s$ is an $\mathbb H$-valued state.
\end{proof}

We denote by $\mathcal S_\mathbb H(E)$ and $\mathcal D_\mathbb H(E)$ the set of all $\mathbb H$-decompositions and the set of all $\mathbb H$-valued states on $H,$ respectively.

\begin{theorem}\label{th:3.2}
There exists a one-to-one correspondence between $\mathcal D_\mathbb H(E)$ and $\mathcal S_\mathbb H(E).$
\end{theorem}

\begin{proof} By Theorem \ref{th:3.1}, $\mathcal D_\mathbb H(E)$ is nonempty iff so is $\mathcal S_\mathbb H(E).$  Assume that $D=(E_t: t \in [0,1]_\mathbb H)$ is an $\mathbb H$-decomposition of $E.$ We define a mapping $f: \mathcal D_\mathbb H(E)\to \mathcal S_\mathbb H (E)$ by $f(D) =s$, where $s(x)=t$ iff $x \in E_t$ for $t \in [0,1]_\mathbb H.$  By the proof of Theorem \ref{th:3.1}, $f$ is bijective.
\end{proof}

\begin{remark}\label{re:3.4}
{\rm If in the definition of an $\mathbb H$-decomposition $(E_t: t\in \mathbb [0,1]_\mathbb H)$ of $E$ we do not suppose that every $E_t$ has to be nonempty, then in the same way as it was proved in  Theorem  \ref{th:3.2} we can prove a one-to-one correspondence between these new $\mathbb H$-decompositions and $\mathbb H$-states, i.e. we suppose only that $s(E) \subseteq [0,1]_\mathbb H.$
}
\end{remark}

In what follows, we will assume that in any decomposition $(E_t: t \in [0,1]_\mathbb H)$ every $E_t$ is nonempty.

\begin{corollary}\label{co:3.3} Let $(E_t: t \in [0,1]_\mathbb H)$ be an $\mathbb H$-decomposition of a pseudo effect algebra $E.$ Then $E_0$ is a normal ideal of $E$ and $E_0 = \Infinit(E).$
\end{corollary}

\begin{proof}
By Theorem \ref{th:3.1}, there exists a unique  $\mathbb H$-valued state $s$ on $E$ such that $s(x)=t$ iff $x \in E_t,$ $t \in [0,1]_\mathbb H.$  Hence, $E_0 =\Ker(s)$ and $\Ker(s)$ is always a normal ideal. Since $E_0 +E_0$ is defined in $E,$ we have $E_0+E_0=E_0 \subseteq \Infinit(E).$ Let now $x\in \Infinit(E),$ then $mx$ is defined in $E$ for any integer $m\ge 1.$  Hence, $s(mx)=m s(x)\le 1,$ so that $s(x)= 0$ and $x \in E_0.$
\end{proof}

We say that an $\mathbb H$-decomposition $(E_t: t \in [0,1]_\mathbb H)$ of $E$ is {\it ordered} if, for $s<t,$ $s,t \in [0,1]_\mathbb H,$ we have $E_s \leqslant E_t.$

\begin{theorem}\label{th:3.4}
Let $(E_t: t \in [0,1]_\mathbb H)$ be an $\mathbb H$-decomposition of a pseudo effect algebra $E$. Then $(E_t: t \in [0,1]_\mathbb H)$ is ordered  if and only if $E_{s}+ E_{t}$ exists in $E$
whenever $s+t<1$ for  $s,t \in [0,1]_\mathbb H.$

In such a case,
\begin{itemize}

\item[{\rm (i)}] $E_{0}=\Infinit(E)$ and $\Infinit(E)$ is a normal ideal.

\item[{\rm (ii)}] $E_{s}+ E_{t}=E_{s+t}$ whenever $s+t<1$.

\item[{\rm (iii)}] If $s+t>1$, for any $x\in E_{s}$ and $y\in E_{t},$
neither $x+ y$ nor $y+ x$ exists.
\end{itemize}

\end{theorem}

\begin{proof} By Theorem \ref{th:3.1}, there is a unique  $\mathbb H$-valued state $s$ such that $s(E_t)=t$ for each $t \in [0,1]_\mathbb H.$

Assume $(E_t: t \in [0,1]_\mathbb H)$ is ordered. Choose $u,v\in [0,1]_\mathbb H$ with $u+v<1.$ We have that $u<1-v$, and  $E_{u}\leqslant E_{1-v}=E_{v}^{-},$ which implies that $E_{u}+
E_v$ exists and we show $E_{u}+ E_{v}=E_{u+v}.$ Indeed, for any $a\in E_{u},$ and any $b\in E_{v},$  we have $s(a+ b)=u+v$,
which implies that $a+ b\in E_{u+v}.$ Conversely, let $c\in E_{u+v}.$ For any $a\in E_{u}$, we have that $a\le c.$ Then there exists an element  $b\in E$ such that $a+ b=c.$ Hence, $s(a+ b)=s(a)+s(b)=u+v,$ then $s(b)=v$, which implies that $b\in E_{v}.$

Now let $E_t + E_w$ exist in $E$ for $t+w<1,$ $t,w \in [0,1]_\mathbb H.$ Choose $u <v,$ $u,v \in [0,1]_\mathbb H.$ Then $u +(1-v)<1$ so that $E_u + E_{1-v}=E_u+E_v^-$ exists in $E$ which yields $E_u \leqslant E_v.$

(i) For any $x, y\in E_{0},$ we have that $x+ y$ exists in $E.$  Then $s(x+y)=s(x)+s(y)=0$ and $x+y \in E_0$ which implies
$E_{0}\subseteq \Infinit(E).$
Conversely, let $x\in \Infinit(E),$ we have that $mx$  is defined in $E$ for each integer $m\ge 1.$ Then $s(mx)=ms(x) \leqslant 1$ which implies $s(x)=0$ and $x\in \Ker(s),$ and so $x\in E_{0}.$

(ii) Assume that $a\in E_{u}$ and $b\in E_{v}$ for $u+v<1.$ Then
$a+ b$ exists and $s(a+ b)= u+v,$ and so
$a+ b\in E_{u+v}.$ Conversely, let $z\in E_{u+v},$ then for any
$x\in E_{u}$, $x\le z$, so that $z=x+(x/z),$ by
$s(z)=s(x)+s(x/z),$ which implies that $x/z\in E_{v}.$

(iii) Assume that $u+v>1$, $x\in E_{u},$ $y\in E_{v},$ either
$x+ y$  or $y+ x$ exists, then $s(x+ y)>1$ or
$s(y+ x)>1$, which is absurd.
\end{proof}

We recall the following two definitions of radicals used in \cite{Dv08}. Let $E$ be a pseudo effect algebra. We denote by $\mathcal{M}(E)$ and $\mathcal{N}(E)$ the set of
maximal ideals and the set of normal ideals of $E,$ respectively. We define (i) the {\it radical} of  $E$, $\Rad(E)$, as the set
$$
\Rad(E)=\bigcap\{I: I\in \mathcal{M}(E)\},
$$
and (ii) the {\it normal radical} of $E$, via
$$
\Rad_{n}(E)=\bigcap\{I: I\in \mathcal{M}(E)\cap\mathcal{N}(E)\}.
$$
It is obvious that $\Rad(E)\subseteq \Rad_{n}(E)$ holds in any pseudo effect algebra
$E.$

\begin{lemma}\label{le:3.5}
{\rm (1)} Let $(E_t: t \in [0,1]_\mathbb H)$ be an ordered $\mathbb H$-decomposition of a pseudo effect algebra $E.$  Then $E_0$ is a unique maximal ideal of $E,$ and it is a normal ideal such that
$E_{0}=\Infinit(E)=\Rad(E)=\Rad_{n}(E).$

{\rm (2)} If $(E_t: t \in [0,1]_\mathbb H)$ is any $\mathbb H$-decomposition of $E$ of type I, then $E_0 + E_0 = E_0.$
\end{lemma}

\begin{proof}
(1) By Corollary \ref{co:3.3}, $E_0$ is a normal ideal, such that $E_{0}=\Infinit(E)$. Now we show that $E_0$ is a maximal ideal. Take $x \in E_t \setminus E_0,$ where $0 < t <1,$ $t \in [0,1]_\mathbb H.$ Let $I$ be an ideal of $E$ generated by $E_0$ and $x.$
Then, for every $s <t,$ $s \in [0,1]_\mathbb H,$ we have $E_s \leqslant  E_t,$ whence $E_s \subseteq I.$ There are two cases: (i) there is no $s\in [0,1]_\mathbb H$ such that $0<s<t.$ Then $t= 1/n$ for some integer $n\ge 1$ and $\mathbb H =\frac{1}{n}\mathbb Z.$  If $n=1$, then $s(x)=1,$  $s(x^-) =0,$ and $x^- \in E_0.$ Hence, $1 \in I.$

If $n \ge 2,$ then $y:=(n-1)x$ is defined in $E,$ and $y \in I.$ For the element $y^-$ we have $s(y^-) = 1/n,$ so that $y^- \in I$ which means $1 \in I.$

(ii) $\mathbb H$ is no cyclic subgroup, so that it is dense in $\mathbb R.$ There is a strictly decreasing sequence $\{t_i\}$ of non-zero elements of $[0,1]_\mathbb H$ such that $t_i\searrow 0.$ For every $t_i$ there is a maximal integer $m_i$ such that $y_i:=m_i t_i$ is defined in $E.$ Hence, for enough small $t_i$, $s(y_i^-) <t$ so that $y_i^- \in I$ which again proves $I=E,$ and $E_0$ is a maximal ideal.

{\it Uniqueness.} Assume that $I$ is another maximal ideal of $E.$  If there is $x \in E_t \cap I$ for some $t \in [0,1]_\mathbb H,$ $t>0,$ then, for every $z \in M_0,$ we have $z \le x$ and $z \in I,$ so that $E_0 \subseteq I.$ The maximality of $E_0$ yields $E_0=I.$

Finally, we have $\Rad(E)=\Rad_{n}(E).$

(2) By (Ii), $E_0+E_0$ is defined in $E,$ and by (Iii), $E_0$ is an ideal. Hence, $E_0 = E_0 +\{0\} \subseteq E_0 + E_0 \subseteq E_0.$
\end{proof}

\begin{proposition}\label{pr:4.1}
Let $G$ be a directed po-group with \RDP$_1$ and let $\mathbb H$ be a subgroup of $\mathbb R$ such that $1 \in \mathbb H.$ Choose $g_0 \in$ G and  set $E_0=\{(0,g): g \in G^+\},$ $E_1=\{(1,g): g \le g_0\}$ and $G_t =\{(t,g): g \in G\}$ for $t \in [0,1]_\mathbb H,$ $0<t<1.$ Then $E:=\Gamma(\mathbb Z \lex G,(1,g_0))$ is a  pseudo effect algebra which is $\mathbb H$-perfect, $(E_t: t \in [0,1]_\mathbb H)$ is an ordered $\mathbb H$-decomposition. $E$ is symmetric if and only if $g_0 \in C(G).$

In addition, $E$ has a unique state, it is an $\mathbb H$-valued state, and $E$ has a unique $\mathbb H$-decomposition.
\end{proposition}

\begin{proof}
It is clear that $(E_t: t \in [0,1]_\mathbb H)$ is an ordered $\mathbb H$-decomposition of $E,$ and $E$ is a symmetric $\mathbb H$-perfect pseudo effect algebra iff $g_0 \in C(G),$ see \cite[p. 98]{187}.

By Theorem \ref{th:3.1}, $E$ has an $\mathbb H$-valued state, say $s,$ namely $s(E_t)=\{t\},$ $t \in [0,1]_\mathbb H.$  Assume that $s_1$ is any state on $E.$  Then $E_0 = \Infinit(E)\subseteq \Ker(s_1),$ and since $E_0$ is a maximal ideal, $E_0=\Ker(s_1).$

Since $\mathbb H$ is in fact a linearly ordered $\ell$-group, we can define the MV-algebra $M:=\Gamma(\mathbb H,1),$ for more detail on MV-algebras, see e.g. \cite{CDM}. Let $\hat s_1$ be a mapping on $M$ defined by $\hat s_1(t)=s_1(t,0)$ if $t \in [0,1]_\mathbb H \setminus \{1\}$ and $s_1(1)=1=s_1(1,g_0).$ It is straightforward to verify that $\hat s_1$ is a state on $M.$ In the same way we can define also a state $\hat s$ on $M.$  Moreover, $M$ being linearly ordered, $M$ has a unique state, see \cite[Thm 5.5]{156},  and $\hat s_1 =\hat s.$

By Theorem \ref{th:2.8}, $\mathbb H\lex G$ has RDP$_1,$ therefore, $s_1$ can be uniquely extended to a state $\overline {s_1}$ on the unital $\ell$-group $(\mathbb H\lex G, (1,g_0)).$

We have $s_1(t,0)=t$ for $t \in  [0,1]_\mathbb H \setminus \{1\}$ and $s_1(1,g_0).$ Therefore, $s_1(0,g)=0$ for any $g \in G^+$ and since $G$ is directed, $\overline{s_1}(0,g)=0$ for any $g \in G.$ If $0<t < 1,$ then $s_1(t,g) = \overline{s_1}(t,g)= s_1(t,0) + \overline{s_1}(0,g) = t.$ On the other hand, if $(1,g) \in E_1,$ then $(1,g)^- \in E_.$ Therefore, $0=s_1((1,g)^-)$ which yields $s_1(1,g)=1.$ In other words, $s = s_1.$

Applying Theorem \ref{th:3.2}, we see  $E$ has a unique $\mathbb H$-decomposition.
\end{proof}

\section{$\mathbb H$-perfect, $\mathbb R$-perfect and $\mathbb Q$-perfect Pseudo Effect Algebras}

We say that a pseudo effect algebra $E$ is $\mathbb H$-{\it perfect} if it has an ordered $\mathbb H$-decomposition $(E_t: t \in [0,1]_\mathbb H).$

In particular, if $\mathbb H = \mathbb R$ or $\mathbb H =\mathbb Q,$ we are saying that $E$ is $\mathbb R$-perfect and $\mathbb Q$-perfect, respectively. If $\mathbb H = \frac{1}{n}\mathbb Z,$ a $\frac{1}{n}\mathbb Z$-perfect pseudo effect algebra is said to be also an $n$-{\it perfect pseudo effect algebra}, see \cite{DXY}.

The following notion of a cyclic element was defined for a special class of pseudo effect algebras, called GMV-algebras, in \cite{Dv08, 225} and for pseudo effect algebras in \cite{DXY}.

Let $n\ge 1$ be an integer. An element $a$ of a pseudo effect algebra $E$ is said to be {\it cyclic of order} $n$ or simply {\it cyclic} if $na$ exists in $E$ and $na =1.$ If $a$ is a cyclic element of order $n$, then $a^- = a^\sim$, indeed, $a^- = (n-1) a = a^\sim$. It is clear that $1$ is a cyclic element of order $1.$

We note that a pseudo effect algebra $E$ has a cyclic element of order $n$ iff $E$ has a pseudo effect subalgebra of $E$ isomorphic to $\Gamma(\frac{1}{n}\mathbb Z,1).$

\begin{proposition}\label{le:4.1}
If $c$ is a cyclic element, then, for any $x\in E,$ $x+ c$ exists in $E$ if and only if $c+ x$ exists in $E.$
\end{proposition}

\begin{proof}
There is an integer $n \ge 1$ such that  $nc=1.$  We have that $c^{\sim}=c^{-}$. Then $x+ c$
exists if and only if $x\le c^{-}$ if and only if
$x\le c^{\sim}$ if and only if $c+ x$ exists.
\end{proof}

We say that a group $G$ is {\it torsion-free} if $ng\ne 0$ for any $g\ne 0$ and every nonzero integer $n.$  For example, every $\ell$-group is torsion-free, see \cite[Cor 2.1.3]{Gla}. We recall that  a po-group $G$ is  torsion-free iff  $\mathbb H \lex G$ is torsion-free, where $\mathbb H$ is a subgroup of the group $\mathbb R$ with $1\in \mathbb H.$

We recall that a group $G$ enjoys {\it unique  extraction of
roots} if, for all positive integers $n$ and all $g,h \in G$, $ng =
nh$ implies $g=h$.  We note that every linearly ordered group, or
a representable $\ell$-group (i.e. it is a subdirect product of linearly ordered groups), in particular every Abelian
$\ell$-group, enjoys unique  extraction of roots, see \cite[Lem.
2.1.4]{Gla}. On the other hand, there is a unital divisible $\ell$-group (i.e. $g/n$ is defined in $G$ for every $g \in G$ and every integer $n\ge 1$), but the unique extraction of roots fails, see \cite[p. 16]{Gla}.

The following results have been originally proved in \cite{DXY} for strong $n$-perfect pseudo effect algebras. In what follows, we generalize them for $\mathbb H$-perfect effect algebras.

Let $E=\Gamma(G,u)$ for some unital po-group $(G,u).$ An element
$c\in E$ such that (a) $nc=u$ for some integer $n \ge 1,$ and (b) $c\in   C(H),$ 
is said to be a {\it strong cyclic element of order $n$.}

For example, if $E = \Gamma(\mathbb R \lex G, (1,0)),$ then any element of the form $(1/n,0),$ $n \ge 1,$ is a strong cyclic element of order $n.$

We note that this notion was used in \cite{Dv08,DXY} to show when a pseudo effect algebra is of the form $\Gamma(\frac{1}{n} \mathbb Z \lex G,(1,0)).$

As a matter of interest, we recall that if $\mathbb H= \mathbb H(\alpha)$ is a subgroup of $\mathbb R$ generated by $\alpha \in [0,1]$ and $1,$ then $\mathbb H = \frac{1}{n}\mathbb Z$ for some integer $n\ge 1$ if $\alpha$ is a rational number. Otherwise, $\mathbb H(\alpha)$ is countable and dense in $\mathbb R,$ and $E(\alpha):= \Gamma(\mathbb H(\alpha),1)=\{m+n\alpha: m,n \in \mathbb Z,\ 0\le m+n\alpha \le 1\},$ see \cite[p. 149]{CDM}. If $\alpha$ is irrational,  then $E(\alpha)$ and $\Gamma(\mathbb H(\alpha)\lex G, (1,0))$ have only one (strong) cyclic element, namely $1$ which is of order $1.$

Now let $G$ be a non Abelian po-group and choose an element $g_0 \in G \setminus C(G).$ Set $E_t:=\{(t,g)\in \Gamma(\mathbb H \lex G, (1,g_0))\},$ $t \in [0,1]_\mathbb H.$ Then $E:=\Gamma(\mathbb  H \lex  G, (1,g_0))$ is $\mathbb H$-perfect, $(E_t: t \in [0,1]_\mathbb H)$ is an ordered $\mathbb H$-decomposition of $E$ but $E$ does not have any strong cyclic element. We note that $E$ is not symmetric.

\begin{lemma}\label{le:4.2}
Let $H$ be a torsion-free po-group with a strong unit $u.$ Let $c\in E=\Gamma(H,u)$ be a strong cyclic element  of order $n.$ If $d\in E$ is any cyclic element of order $n,$ then $c=d.$
\end{lemma}

\begin{proof}
Since $c \in C(H)$ and $d\in H,$ we have $c+d=d+c$ in the group $H.$ Then $n(c-d)=nc-nd=0$ so that $c=d.$
\end{proof}

\begin{proposition}\label{pr:4.3}
Let  $E$ be a  pseudo effect algebra with an ordered $\mathbb H$-decomposition $(E_t: t \in [0,1]_\mathbb H).$ Then there exists a unique directed  po-group $G$ such that $G^{+}=E_{0}$.
\end{proposition}

\begin{proof}
By Proposition \ref{le:3.5}(1), $E_0$ is a maximal ideal of $E,$ and $E_0+E_0$ is defined, so that $E_0+E_0=E_0.$ Hence, $(E_{0};+,0)$ is a semigroup.
For any $x,y\in E_{0},$ the equation $x+y=0,$ implies that $x=y=0.$
For any $x,y,z\in E_{0},$ the equation $x+y=x+z$ implies that $y=z$,
and equation $y+x=z+x$ implies that $y=z$. Then $(E_{0};+,0)$  is a
cancellative semigroup satisfying the conditions of the Birkhoff Theorem, \cite[Thm II.4]{Fuc}, which guarantees that $E_{0}$ is the positive
cone of a unique (up to isomorphism) po-group $G$. Without loss of generality, we can assume that $G$ is generated by the positive cone $E_0,$ so that $G$ is directed.
\end{proof}

We say that a pseudo effect algebra $E$ enjoys the $1$-{\it divisibility property} if, given integer $n\ge 1,$ there is an element $a_n\in E$ such that $na_n =1.$ We see that $a_n$ is a cyclic element of order $n.$  We recall that $a_n$ is not necessarily unique, and we denote $a_n=\frac{1}{n}1$ if $a_n$ is unique. For example by \cite[p. 16]{Gla}, there is a unital $\ell$-group $(G,u)$ such that in $\Gamma(G,u)$ there are two different elements $a\ne b$ with $2a=u= 2b.$

Let $E=\Gamma(G,u)$ be an interval pseudo effect algebra. We say that $E$ enjoys the {\it strong}  $1$-{\it divisibility property} if, given integer $n \ge 1$, there is an element $a_n \in C(G)\cap E$ such that $na_n=1.$ We see that $a_n$ is a strong cyclic element of order $n.$ If $G$ is e.g. torsion-free, then due to Lemma \ref{le:4.2}, $a_n$ is unique.

We say that a pseudo effect algebra $E$ enjoys
{\it unique  extraction of roots of $1$} if $a,b \in E$ and $n a, nb$ exist in $E$, and $n a=1= n b$, then $a= b.$  Then every $\Gamma(\mathbb H \lex G,(1,0)) $ enjoys unique  extraction of roots of $1$ for any $n\ge 1$
and for any torsion-free directed po-group $G$.  Indeed, let $k(s,g) = (1,0)=k(t,h)$ for some $s,t \in [0,1]_\mathbb H,$ $g,h \in G$, $k\ge 1.$ Then $ks=1=kt$ which yields $s=t >0$, and $kg=0=kh$ implies $g=0=h.$

We say that  a nonempty subset $A$ of a pseudo effect algebra $E$ is (i) {\it downwards directed} if given $a_1,a_2\in A,$ there is an element $a\in A$ such that $a \le a_1,a_2;$ (ii) {\it upwards directed} if given $a_1,a_2\in A,$ there is an element $b\in A$ such that $b \ge a_1,a_2;$ and (iii) {\it directed} if it is both upwards and downwards directed.

\begin{proposition}\label{pr:4.4}
{\rm (1)} Let $(E_t: t \in [0,1]_\mathbb H)$ be an ordered $\mathbb H$-decomposition of a pseudo effect algebra $E.$ Then $E_0$ and $E_1$ are directed.

{\rm (2)} If $(E_0,E_{1/n},\ldots,E_{n/n})$ is an ordered $\frac{1}{n}\mathbb Z$-decomposition of a pseudo effect algebra $E$ satisfying \RDP$_0,$  then every $E_i$ is directed, $i=0,1,\ldots,n.$
\end{proposition}

\begin{proof}
(1) It is clear that $E_0$ is downwards directed. Let $a,b \in E_0.$ Since $E_0=\Infinit(E),$ $a+b$ is defined in $E$ and $a+b \in E_0.$ Then $a,b \le a+b.$ Since $E_1=E_0^-,$ we see that $E_1$ is directed.

(2) It follows from \cite[Prop 5.12]{DXY}.
\end{proof}

We recall that if we set $E=\Gamma(\mathbb H \lex G,(1,0))$ for some po-group $G,$ then $E_t$ for $t \in [0,1]\setminus\{0,1\}$ is neither downwards nor upwards directed whenever $G$ is not directed, and in general, $E_t$ is directed iff $G$ is directed.

Inspiring by this, we say that an $\mathbb H$-decomposition $(E_t: t\in [0,1]_\mathbb H)$ of $E$ has the {\it directness property} if every $E_t$ is directed.

\section{Representation of Strong $\mathbb Q$-perfect Pseudo Effect Algebras}

In this section, we define a strong $\mathbb Q$-perfect pseudo effect algebra and we show when it is an interval in the lexicographic product $\mathbb Q \lex G.$ In addition, we derive a categorical equivalence of the category of strong $\mathbb Q$-perfect pseudo effect algebras with the category of directed torsion-free po-groups.

\begin{definition}\label{de:5.1}
{\rm We say that a pseudo effect algebra $E$ with RDP$_1$ is  {\it strong} $\mathbb Q$-{\it perfect} if
\begin{enumerate}
\item[(i)] $E$ enjoys the strong $1$-divisibility property,

\item[(ii)] $E$ has an ordered $\mathbb Q$-decomposition having the directness property,

\item[(iii)] the unital po-group $(G,u)$ such that $E = \Gamma(G,u)$  is torsion free.
\end{enumerate}
}
\end{definition}

In what follows, we present a representation of strong $\mathbb Q$-perfect pseudo effect algebras. We start with a preparatory statement.

\begin{proposition}\label{pr:5.2}
Let $G$ be a directed torsion-free  po-group with {\rm RDP}$_1$. Then  the pseudo effect algebra $$
\mathcal Q(G):=\Gamma(\mathbb Q \lex G,(1,0)) \eqno(6.1)
$$
is a strong $\mathbb Q$-perfect pseudo effect algebra.
\end{proposition}

\begin{proof}
By Theorem \ref{th:2.5}, $\mathcal Q(G)$ has RDP$_1$, it has  an ordered $\mathbb Q$-decomposition  $(E_t: t \in [0,1]_\mathbb Q),$ where $E_0=\{(0,g): g \in G^+\},$ $E_1=\{(1,-g): g \in G^+\},$ and $E_t=\{(t,g): g \in G\},$ for $t\in [0,1]_\mathbb Q \setminus\{0,1\},$  and each element $a_n :=(1/n,0)$ is a strong cyclic element of order $n.$ Thus $E$ enjoys the strong $1$-divisibility property, and $G$ being directed, every $E_t$ is directed.
\end{proof}

\begin{theorem}\label{th:5.3}
Let $E$ be a strong $\mathbb Q$-perfect pseudo effect algebra with {\rm RDP}$_1.$  Then there is a unique (up to isomorphism) torsion-free directed po-group $G$ with {\rm RDP}$_1$ such that $E \cong \Gamma(\mathbb Q \lex G,(1,0)).$
\end{theorem}

\begin{proof}
Since $E$ has RDP$_1,$ due to \cite[Thm 5.7]{DvVe2}, there is a unique unital (up to isomorphism of unital po-groups)  po-group $(H,u)$ with RDP$_1$ such that $E = \Gamma(H,u).$ Assume $(E_t: t \in [0,1]_\mathbb Q)$ is an ordered $\mathbb Q$-decomposition of $E$ with the directness property; due to Theorem \ref{th:3.2}, it is unique. By Theorem \ref{th:3.1}, there is a unique $\mathbb Q$-valued state $s.$ By Proposition \ref{pr:4.3}, there is a unique directed po-group $G$ such that $E_0=G^+.$ Since $H$ is torsion-free and with RDP$_1$, $G=G^+-G^+$ is also torsion-free, directed and with RDP$_1.$ For any integer $n\ge 1,$ there is a unique element $a_n=\frac{1}{n}1.$ It is clear that, for any integer $1\le m\le n,$ the element $m\frac{1}{n}1$ is defined and we denote it by $\frac{m}{n}1.$  Define $\mathcal Q(G)$ by (4.1) and define a mapping $\phi: E \to \mathcal Q(G)$ by

$$
\phi(x) =\textstyle{(\frac{i}{n}, x-\frac{i}{n}1)}\eqno(6.2)
$$
whenever $ x\in E_{\frac{i}{n}}.$ We note that ``$-$" in the right-hand side of the formula means the group subtraction in the po-group $H.$ Since $E_{\frac{i}{n}}$ is directed, there is an element $x_0 \in E_{\frac{i}{n}}$ such that $x_0 \le x, \frac{i}{n}1.$ Then $x-x_0 = x \backslash x_0 \in E_0$ and $\frac{i}{n}1 - x_0=\frac{i}{n}1 \backslash x_0 \in E_0$ because $s(x \backslash x_0)=0= s(\frac{i}{n}1 \backslash x_0).$ Therefore, $x - \frac{i}{n}1 = (x-x_0)-(\frac{i}{n}1-x_0) \in G,$ which means that $\phi$ is a well-defined mapping.

\vspace{2mm}

{\it Claim:} {\it The mapping $\phi$ is an injective and surjective homomorphism of pseudo effect algebras.}

\vspace{2mm}

We have $\phi(0)=(0,0)$ and $\phi(1)=(1,0).$ Let $x \in E_\frac{i}{n}.$ Then $x^- \in E_\frac{n-i}{n},$ and $\phi(x^-) =(\frac{n-i}{n}, x^- - \frac{n-i}{n}1) = (1,0)-(\frac{i}{n},x - \frac{i}{n}1)=\phi(x)^-.$ In an analogous way, $\phi(x^\sim)=\phi(x)^\sim.$

Now let $x,y \in E$ and let $x+y$ be defined in $E.$ Then $x\in E_\frac{i}{n_1}$ and $y \in E_\frac{j}{n_2}.$ Without loss of generality, we can assume that $n_1=n_2 =n.$
Since $x\le y^-,$ we have $\frac{i}{n} \le \frac{n-j}{n}.$ So that $\phi(x) \le \phi(y^-)=\phi(y)^-$ which means  $\phi(x)+\phi(y)$ is defined in $\mathcal Q(G).$ Then $\phi(x+y) = (\frac{i+j}{n}, x+y - \frac{i+j}{n}1) =
(\frac{i+j}{n}, x+y -(\frac{i}{n}1 + \frac{j}{n}1))= (\frac{i}{n},x-\frac{i}{n}1) + (\frac{j}{n},y- \frac{j}{n}1)=\phi(x)+\phi(y).
$

Assume $\phi(x)\le \phi(y)$ for some $x\in E_\frac{i}{n}$ and $y \in E_\frac{j}{n}.$ Then $(\frac{i}{n},x-\frac{i}{n}1)\le (\frac{j}{n}, y - \frac{i}{n}1).$ If $i=j,$ then $x-\frac{i}{n}1\le y-\frac{i}{n}1$ so that $x\le y.$  If $i<j,$ then $x \in E_\frac{i}{n}$ and $y\in E_\frac{j}{n}$ so that $x<y.$  Therefore, $\phi$ is injective.

To prove that $\phi$ is surjective, assume two cases: (i) Take $g \in G^+=E_0.$  Then $\phi(g)=(0,g).$ In addition $g^- \in E_1$ so that $\phi(g^-) = \phi(g)^-= (0,g)^- = (1,0)-(0,g)=(1,-g).$ (ii) Let $g \in G$ and $\frac{i}{n}$ with $1<i<n$ be given. Then $g = g_1-g_2,$ where $g_1,g_2 \in G^+=E_0.$ Since $\frac{i}{n}1 \in E_\frac{i}{n},$ $g_1 + \frac{i}{n}1$ exists in $E$ and it belongs to $E_\frac{i}{n},$ and $g_2 \le g_1+\frac{i}{n}1$ which yields $(g_1+\frac{i}{n}1)- g_2 = (g_1+\frac{i}{n}1)\backslash g_2 \in E_\frac{i}{n}.$  Hence, $g+\frac{i}{n}1 = (g_1 + \frac{i}{n}1)\backslash g_2 \in E_\frac{i}{n}$ which entails $\phi(g+\frac{i}{n}1)=(\frac{i}{n},g).$

Consequently, $E$ is isomorphic to $\mathcal Q(G).$

If $E \cong \Gamma(\mathbb Q\lex G',(1,0)),$ then $G$ and $G'$ are isomorphic po-groups.
\end{proof}

We note that the notion of the directness property of a pseudo effect algebra $E$ was introduce in order to guarantee the existence of the element $x- \frac{i}{n}1$ in the po-group $G,$ and consequently, to show that the mapping $\phi$ in (6.2) is well-defined.  If $E$ satisfies RDP$_2,$ then by \cite{DvVe1,DvVe2}, $E$ is in fact a pseudo MV-algebra, and by  \cite{151}, there is a unital $\ell$-group $(H,u)$ such that $E =\Gamma(H,u).$ At any rate, any $\ell$-group is torsion-free and directed.  Hence, let for this case assume that a pseudo effect algebra $E$ with RDP$_2$ is strong $\mathbb Q$-perfect if it enjoys  the strong $1$-divisibility property and it has an ordered $\mathbb Q$-decomposition. We assert that $\phi$ is well-defined. Indeed, $E_0$ in this case is in fact a positive cone of an $\ell$-group $G.$  The unique state $s$ on $E$ corresponding to the given ordered $\mathbb H$-decomposition is an extremal state, therefore, by \cite[Prop 4.7]{156}, $s(a\wedge b)=\min\{s(a),s(b)\}$ for all $a,b \in E,$ and the same is true for its extension $\hat s$ onto $(H,u)$ and all $a,b \in H.$  Let $x \in E_\frac{i}{n}.$ For the element $x - \frac{i}{n}1 \in H,$ we define $(x-\frac{i}{n}1)^+:= (x-\frac{i}{n}1)\vee 0 = (x \vee \frac{i}{n}1)-\frac{i}{n}1 \in E_0$ while $s((x \vee \frac{i}{n}1)-\frac{i}{n}1)=s(x \vee \frac{i}{n}1)-s(\frac{i}{n}1)= \frac{i}{n}-\frac{i}{n}=0$ and similarly $(x -\frac{i}{n}1)^- := -((x-\frac{i}{n}1)\wedge 0) = \frac{i}{n}1 - (x\wedge \frac{i}{n}1) \in E_0.$ This implies that $x-\frac{i}{n}1= (x-\frac{i}{n}1)^+ - (x-\frac{i}{n}1)^-\in G.$ It proves the following result.

\begin{theorem}\label{th:5.4}
Let $E$ be a strong $\mathbb Q$-perfect pseudo effect algebra in the sense  of the last note and  with \RDP$_2.$  Then there is a unique $\ell$-group $G$ such that $E\cong \Gamma(\mathbb Q\lex G,(1,0)).$
\end{theorem}

\begin{proof}
It follows from the last note  and the proof of Theorem \ref{th:5.3}.
\end{proof}


{\bf Problem 2.} Can we relax the notion of directness in the definition of a strong $\mathbb Q$-perfect pseudo effect algebra in order to prove Theorem \ref{th:5.3}~?


In what follows, we show the categorical equivalence of the category of strong $\mathbb Q$-perfect pseudo effect algebras with the category of directed torsion-free po-groups with RDP$_1.$

Let $\mathcal {SQPPEA}$ be the category of strong $\mathbb Q$-perfect pseudo effect algebras whose objects are strong $\mathbb Q$-perfect pseudo effect algebras and morphisms are homomorphisms of pseudo effect algebras. Similarly, let $\mathcal G$ be the category whose objects are directed torsion-free po-groups  with RDP$_1,$ and morphisms are homomorphisms of unital po-groups.

Define a functor $\mathcal E_\mathbb Q: \mathcal G \to  \mathcal {SQPPEA}$ as follows: for $G\in \mathcal G,$ let
$$
\mathcal E_\mathbb Q(G):= \Gamma(\mathbb Q\lex G,(1,0))
$$
and if $h: G \to G_1$ is a po-group homomorphism, then

\begin{center} $\mathcal E_\mathbb Q(h)(\frac{i}{n},g)= (\frac{i}{n}, h(g)), \quad (\frac{i}{n},g) \in \Gamma(\mathbb Q\lex G,(1,0)).$
\end{center}

By Proposition \ref{pr:5.2}, $\mathcal E_\mathbb Q$ is a well-defined functor.

\begin{proposition}\label{pr:5.5}
$\mathcal E_\mathbb Q$ is a faithful and full
functor from the category ${\mathcal G}$ of directed torsion-free po-groups with \RDP$_1$  into the
category $\mathcal{SQPPEA}$ of strong $\mathbb Q$-perfect pseudo effect algebras.
\end{proposition}

\begin{proof}
Let $h_1$ and $h_2$ be two morphisms from $G$
into $G'$ such that $\mathcal E_\mathbb Q(h_1) = \mathcal E_\mathbb Q(h_2)$. Then
$(0,h_1(g)) = (0,h_2(g))$ for any $g \in G^+$, consequently $h_1 =
h_2.$

To prove that $\mathcal E_\mathbb Q$ is a full  functor, suppose that $f$ is a morphism from a strong $\mathbb Q$-perfect pseudo effect algebra
$\Gamma(\mathbb Q\lex G, (1,0))$ into another one $\Gamma(\mathbb
Q\lex G_1, (1,0)).$  Then $f(0,g)
= (0,g')$ for a unique $g' \in G'^+$. Define a mapping $h:\ G^+ \to
G'^+$ by $h(g) = g'$ iff $f(0,g) =(0,g').$ Then $h(g_1+g_2) = h(g_1)
+ h(g_2)$ if $g_1,g_2 \in G^+.$
Assume now that $g \in G$ is arbitrary. Then $g = g_1 -
g_2 = g_1'-g_2'$, where $g_1, g_2, g_1', g_2' \in G^+,$ which gives
$g_1 +g_2' = g_1' + g_2$, i.e., $h(g) = h(g_1) - h(g_2)$ is a
well-defined extension of $h$ from $G^+$ onto $G$.

Let $0\le g_1 \le g_2.$ Then $(0,g_1)\le (0,g_2),$
which means  $h$ is a homomorphism of po-groups, and $\mathcal E_\mathbb Q(h)
= f$ as desired.
\end{proof}

We recall that by a {\it universal group}  for a
pseudo effect algebra $E$  we mean a pair $(G,\gamma)$ consisting of a po-group $G$ and a $G$-valued measure $\gamma :\, E\to G^+$
(i.e., $\gamma (a+b) = \gamma(a) + \gamma(b)$ whenever $a+b$ is
defined in $E$) such that the following conditions hold: (i)
$\gamma(E)$ generates $G$, and  (ii) if $H$ is a group and
$\phi:\, E\to H$ is an $H$-valued measure, then there is a group
homomorphism $\phi^*:{ G}\to H$ such that $\phi ={\phi}^*\circ
\gamma$.

Due to \cite[Thm 7.2]{DvVe2}, every pseudo algebra with RDP$_1$ admits a universal group, which is unique up to isomorphism, and $\phi^*$ is unique. The universal group for $E = \Gamma(G,u)$ is $(G,id)$ where $id$ is the
embedding of $E$ into $G,$ and $G$ satisfies RDP$_1.$

Let $\mathcal A$ and $\mathcal B$ be two categories and let $f:\mathcal A \to \mathcal B$ be a morphism. Suppose that $g,h$ be two morphisms from $\mathcal B$ to $\mathcal A$ such that $g\circ f = id_\mathcal A$ and $f\circ h = id_\mathcal B,$ then $g$ is a {\it left-adjoint} of $f$ and $h$ is a {\it right-adjoint} of $f.$

\begin{proposition}\label{pr:5.6}
The functor $\mathcal  E_\mathbb Q$ from the
category ${\mathcal  G}$ into the category $\mathcal{SQPPEA}$ has  a
left-adjoint.
\end{proposition}

\begin{proof}
We show that given a strong $\mathbb Q$-perfect  pseudo effect algebra $E$ with an ordered $\mathbb Q$-decomposition $(E_t: t \in [0,1]_\mathbb Q),$
there is a universal arrow $(G,f)$, i.e., $G$ is an object in $\mathcal G$ and $f$ is a homomorphism from
$E$ into ${\mathcal  E}_\mathbb Q(G)$ such that if $G'$ is an object from ${\mathcal G}$ and $f'$ is a homomorphism from $E$ into ${\mathcal  E}_\mathbb Q(G')$, then
there exists a unique morphism $f^*:\, G \to G'$ such that ${\mathcal
E}(f^*)\circ f = f'$.

By Theorem \ref{th:5.3}, there is a unique directed torsion-free po-group $G$ with RDP$_1$ such that $E \cong \Gamma(\mathbb Q \lex G,(1,0)).$ By Theorem \ref{th:2.7}, $\mathbb Q \lex G$ is a directed po-group with RDP$_1,$ so that by \cite[Thm 7.2]{DvVe2}, $(\mathbb Q \lex G, \gamma)$ is a universal group for $E,$ where $\gamma: E \to  \Gamma(\mathbb Q \lex G, (1,0))$ is defined by $\gamma(a) = (\frac{i}{n},a -\frac{i}{n}1),$ if $a \in E_\frac{i}{n}.$ By the proof of Theorem \ref{th:5.3}, $\gamma$ is an isomorphism.
\end{proof}

Define a mapping ${\mathcal  P}: \mathcal  {SQPPEA} \to {\mathcal  G}$
via ${\mathcal  P}(E) := G$ whenever $(\mathbb Q\lex  G, f)$ is a
universal group for $E$. It is clear that if $f_0$ is a morphism
from $E$ into $F$, then $f_0$ can be uniquely extended to a
homomorphism ${\mathcal  P}(f_0)$ from $G$ into $G_1$, where $(\mathbb Q
\lex G_1, f_1)$ is a universal group for the strong $\mathbb Q$-perfect
pseudo effect algebra $F$.

\begin{proposition}\label{pr:5.7}
The mapping ${\mathcal  P}$ is a functor from the
category $\mathcal {SQPPEA}$ into the category ${\mathcal  G}$ which is a
left-adjoint of the functor ${\mathcal  E}_\mathbb Q.$
\end{proposition}

\begin{proof}
 It follows from the construction of the
universal group.
\end{proof}

Now we present  the main result on categorical equivalence of the
category of strong $\mathbb Q$-perfect pseudo effect algebras and the category  $\mathcal G.$

\begin{theorem}\label{th:5.8}
The functor ${\mathcal  E}_\mathbb Q$ defines a categorical
equivalence of the category ${\mathcal  G}$   and the
category $\mathcal {SQPPEA}$ of strong $\mathbb Q$-perfect pseudo effect algebras.

In addition, suppose that $h:\ {\mathcal  E}_\mathbb H(G) \to {\mathcal  E}_\mathbb H(H)$ is a
homomorphism of pseudo effect algebras, then there is a unique homomorphism
$f:\ G \to H$ of unital po-groups such that $h = {\mathcal  E}_\mathbb H(f)$,
and
\begin{enumerate}
\item[{\rm (i)}] if $h$ is surjective, so is $f$;
 \item[{\rm (ii)}] if $h$ is  injective, so is $f$.
\end{enumerate}
\end{theorem}

\begin{proof}
According to \cite[Thm IV.4.1]{MaL}, it is
necessary to show that, for a strong $\mathbb Q$-perfect pseudo effect algebra $E$, there is an
object $G$ in ${\mathcal  G}$ such that ${\mathcal  E}_\mathbb Q(G)$ is isomorphic to $E$. To show that, we take a universal group $(\mathbb Q
\lex G, f)$. Then ${\mathcal  E}_\mathbb Q(G)$ and $E$ are isomorphic.
\end{proof}

\section{Strong $\mathbb H$-perfect Pseudo Effect Algebras and Their Representation}

In this section, we extend the results of the previous section to the most general case, namely for strong $\mathbb H$-perfect pseudo effect algebras. Here we use Theorem \ref{th:2.8}.


We say that an $\mathbb H$-decomposition $(E_t: t\in [0,1]_\mathbb H)$ of $E$ has the {\it cyclic property} if there is a system of elements $(c_t\in E: t \in [0,1]_\mathbb H)$ such that (i) $c_t \in E_t$ for any $t \in [0,1]_\mathbb H,$ (ii) if $s+t \le 1,$ $s,t \in [0,1]_\mathbb H,$ then $c_s+c_t=c_{s+t},$ and (iii) $c_1=1.$ Properties: (a) $c_0=0;$ indeed, by (ii) we have $c_0+c_0=c_0,$ so that $c_0=0.$ (b) If $t=1/n,$ then $c_\frac{1}{n}$ is a cyclic element of order $n.$

Let $E =\Gamma(G,u).$ An $\mathbb H$-decomposition $(E_t: t\in [0,1]_\mathbb H)$ of $E$ has the {\it strong cyclic property} if there is a system of elements $(c_t\in E: t \in [0,1]_\mathbb H)$ such that (i) $c_t \in E_t\cap C(G)$ for any $t \in [0,1]_\mathbb H,$ (ii) if $s+t \le 1,$ $s,t \in [0,1]_\mathbb H,$ then $c_s+c_t=c_{s+t},$ and (iii) $c_1=1.$ We recall that if $t=1/n,$ $c_\frac{1}{n}$ is a strong cyclic element of order $n.$

For example, let $E=\Gamma(\mathbb H \lex G,(1,0))$ and $E_t =\{(t,g): (t,g)\in E\}$ for $t \in [0,1]_\mathbb H.$ If we set $c_t =(t,0),$ $t \in [0,1]_\mathbb H,$ then the system $(c_t: t \in [0,1]_\mathbb H)$ satisfies (i)---(iii) of the strong cyclic property, and $(E_t: t \in [0,1]_\mathbb H)$ is an $\mathbb H$-decomposition of $E$ with the strong cyclic property.

Finally, we say that an interval pseudo effect algebra $E$ has the $\mathbb H$-{\it strong cyclic property} if there is a decomposition $(E_t: t \in [0,1]_\mathbb H)$ of $E$ with the strong cyclic property.


\begin{proposition}\label{pr:7.1}
An interval pseudo effect algebra $E =\Gamma(G,u),$ where $G$ is torsion-free, has the $\mathbb Q$-strong cyclic property if and only if $E$ has the strong $1$-divisibility property.
\end{proposition}

\begin{proof}
If a $\mathbb Q$-decomposition $(E_\frac{i}{n}: \frac{i}{n} \in [0,1]_\mathbb Q)$ of $E$  has the strong cyclic property, there is a system of elements $(c_\frac{i}{n}: \frac{i}{n} \in [0,1]_\mathbb Q)$ satisfying (i)--(iii) of the strong cyclic property. Then every $c_\frac{1}{n}$ is a strong cyclic element of order $n$ which by Lemma \ref{le:4.2} means that it is a unique strong cyclic element of order $n.$ Hence, $c_\frac{1}{n}=\frac{1}{n}1.$

Conversely, let $E$ have the strong $1$-divisibility property, then the elements $c_\frac{i}{n}:= \frac{i}{n}1\in E_\frac{i}{n},$ and the system $(c_\frac{i}{n}: \frac{i}{n}\in [0,1]_\mathbb Q)$ satisfies the conditions (i)--(iii) of the strong cyclic property.
\end{proof}

Definition \ref{de:5.1} for strong $\mathbb Q$-perfect pseudo effect algebras will be change for strong $\mathbb H$-perfect pseudo effect algebras as follows.

\begin{definition}\label{de:7.2}
{\rm We say that a pseudo effect algebra $E$ with RDP$_1$ is  {\it strong} $\mathbb H$-{\it perfect} if
\begin{enumerate}

\item[(i)] $E$ has an ordered $\mathbb H$-decomposition having both the directness property and the strong cyclic property,

\item[(ii)] the unital po-group $(G,u)$ such that $E = \Gamma(G,u)$  is torsion free.
\end{enumerate}
}
\end{definition}

Proposition \ref{pr:5.2} has the following counter part.

\begin{proposition}\label{pr:7.3}
Let $G$ be a directed torsion-free  po-group with {\rm RDP}$_1$. Then  the pseudo effect algebra
$$
\mathcal E_\mathbb H(G):=\Gamma(\mathbb H \lex G,(1,0)) \eqno(7.1)
$$
is a strong $\mathbb H$-perfect pseudo effect algebra.
\end{proposition}

The representation theorem for strong $\mathbb H$-perfect pseudo effect algebras by (7.1) is the following result; compare it with Theorem \ref{th:5.3}.

\begin{theorem}\label{th:7.4}
Let $E$ be a strong $\mathbb H$-perfect pseudo effect algebra with {\rm RDP}$_1.$  Then there is a unique (up to isomorphism) torsion-free directed po-group $G$ with {\rm RDP}$_1$ such that $E \cong \Gamma(\mathbb H \lex G,(1,0)).$
\end{theorem}

\begin{proof}
The proof follows the basic steps of the proof of Theorem \ref{th:5.3}. Let $(E_t: t \in [0,1]_\mathbb H)$ be an $\mathbb H$-decomposition with the strong cyclic property with a given system of elements $(c_t: t \in [0,1]_\mathbb H)$ satisfying (i)--(iii) of the strong cyclic property.

Hence, we derive a directed po-group $G$ such that $G^+=E_0.$
Take the $\mathbb H$-strong cyclic pseudo effect algebra $\mathcal E_\mathbb H(G)$ defined by (7.1), and define a mapping $\phi: E \to \mathcal E_\mathbb H(G)$ by

$$
\phi(x):= (t, x - c_t)\eqno (7.2)
$$
whenever $x \in E_t$ for some $t \in [0,1]_\mathbb H,$ where $ x-c_t$ denotes the difference taken in the group $H.$ In the same way as in the proof of Theorem \ref{th:5.3}, we can prove that $\phi$ is a well-defined injective and surjective homomorphism of pseudo effect algebras.
\end{proof}

Finally, let $\mathcal {SPPEA}_\mathbb H$ be the category of strong $\mathbb H$-perfect pseudo effect algebras whose objects are strong $\mathbb H$-perfect pseudo effect algebras and morphisms are homomorphisms of pseudo effect algebras. Again, let $\mathcal G$ be the category whose objects are directed torsion-free po-groups  with RDP$_1,$ and morphisms are homomorphisms of unital po-groups.

Define a functor $\mathcal E_\mathbb H: \mathcal G \to  \mathcal {SPPEA}_\mathbb H$ as follows: for $G\in \mathcal G,$ let
$$
\mathcal E_\mathbb H(G):= \Gamma(\mathbb H\lex G,(1,0))
$$
and if $h: G \to G_1$ is a po-group homomorphism, then

\begin{center} $\mathcal E_\mathbb H(h)(t,g)= (t, h(g)), \quad (t,g) \in \Gamma(\mathbb H\lex G,(1,0)).$
\end{center}

We note that the $\mathcal {SQPPEA}= \mathcal {SPPEA}_\mathbb Q.$

Using the same ideas as for the categorical equivalence in Theorem \ref{th:5.8}, we can prove a new categorical equivalence.

\begin{theorem}\label{th:7.5}
The functor ${\mathcal  E}_\mathbb H$ defines a categorical
equivalence of the category ${\mathcal  G}$   and the
category $\mathcal {SPPEA}_\mathbb H$ of strong $\mathbb H$-perfect pseudo effect algebras.
\end{theorem}

Consequently, if $\mathbb H$ is cyclic or the group of rational numbers or an arbitrary subgroup of reals containing the number 1, all the studied categories of $\mathbb H$-strong perfect pseudo effect algebras are mutually categorically equivalent, and categorical equivalent to the category $\mathcal G.$

\end{document}